\documentclass[12pt, leqno]{article}

\title{
Voros Coefficients and the Topological Recursion 
for a Class of the Hypergeometric Differential Equations 
associated with the Degeneration of the 2-dimensional Garnier System 
}
\author{
Yumiko \textsc{Takei}$^{*}$}
\def\paperinfo{
\renewcommand{\thefootnote}{\fnsymbol{footnote}}
\footnote[0]{$^{*}$
Graduate School of Science and Technology, Kwansei Gakuin University.
 {\tt{ytakei@kwansei.ac.jp}
}
}
\footnote[0]{2010 \textit{Mathematics Subject Classification}. 
Primary:34M60; Secondary:81T45}
\footnote[0]{\textit{Keywords}: Exact WKB analysis; Voros coefficients;
Topological recursion; Quantum curves; Free energy.}
\renewcommand*{\thefootnote}{\arabic{footnote}}
}
\date{\today}

\usepackage{afterpage}
\usepackage{amsmath}
\usepackage{amsthm}
\usepackage{amssymb}        
\usepackage{amsfonts}       
\usepackage[dvipdfmx]{color}
\usepackage[dvipdfmx]{graphicx}
\usepackage{enumitem}
\usepackage[mathscr]{eucal} 
\usepackage[a4paper]{geometry}
\usepackage{lipsum}
\usepackage{lscape}
\usepackage{mathrsfs}
\usepackage{stmaryrd}       
\usepackage[all]{xy}
\usepackage{bm}

\theoremstyle{plain}
\newtheorem{thm}{Theorem}[section]

\newtheorem{lem}[thm]{Lemma}

\theoremstyle{definition}
\newtheorem{dfn}[thm]{Definition}
\theoremstyle{remark}
\newtheorem{rem}[thm]{Remark}
\newtheorem{ex}[thm]{Example}

\numberwithin{equation}{section}
\numberwithin{table}{section}
\numberwithin{figure}{section}

\newcommand{\thmref}[1]{Theorem \ref{thm:#1}}

\newcommand{\lemref}[1]{Lemma \ref{lem:#1}}

\def\Res{\mathop{\rm{Res}}}

\begin{document}

\maketitle

\paperinfo


\vspace{-1.em}
\begin{abstract}
In my joint papers with Iwaki and Koike (\cite{IKT-part1, IKT-part2}) we found an intriguing relation 
between the Voros coefficients in the exact WKB analysis and the free energy
in the topological recursion introduced by Eynard and Orantin in the case of
the confluent family of the Gauss hypergeometric differential equations.
In this paper we discuss its generalization to the case of the hypergeometric
differential equations associated with $2$-dimensional degenerate Garnier systems.
\end{abstract}




\section{Introduction} 
\label{section:intro}


In [IKoT1, IKoT2] we showed that there is an interesting relation 
between the exact WKB theory and the topological recursion for the 
confluent family of the Gauss hypergeometric differential equations, 
that is, we verified that the Voros coefficients of (confluent) 
hypergeometric equations are expressed as the difference values 
of the free energy of the spectral curve obtained as the classical 
limit of the equations. In this paper we discuss the extension of 
this result to a family of hypergeometric differential equations 
associated with $2$-dimensional degenerate Garnier systems. 

The $N$-dimensional Garnier system is a Hamiltonian system with $N$ 
variables obtained through monodromy preserving deformations of second 
order linear differential equations on $\mathbb{P}^1$ with $N+3$ regular 
singular points. In the case of $N=1$, the system reduces to the sixth 
Painlev\'e equation $P_{\rm VI}$ and the Gauss hypergeometric function 
gives a particular solution of $P_{\rm VI}$. In this sense the Gauss 
hypergeometric equation and its confluent version are associated with 
Painlev\'e equations (i.e., $1$-dimensional Garnier system). In the 
same manner a confluent family of hypergeometric differential systems 
with two independent variables are associated with the $2$-dimensional 
Garnier system according to the following diagram of degeneration: 
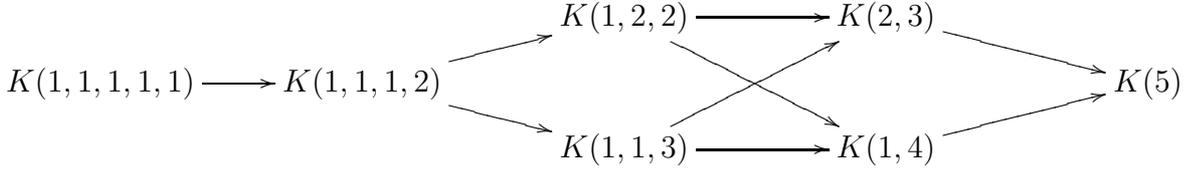
\begin{figure}[htbp]
$$
\xymatrix@!C=25pt@R=7pt{
&&
&& \text{$K(1,2,2)$} \ar@{->}[rr] \ar@{->}[rrdd]
&& \text{$K(2,3)$} \ar@{->}[rrd] 
&&
&&
\\
\text{$K(1,1,1,1,1)$} \ar@{->}[rr] 
&& \text{$K(1,1,1,2)$} \ar@{->}[rru] \ar@{->}[rrd] 
&&
&&
&& \text{$K(5)$} 
\\
&&
&& \text{$K(1,1,3)$} \ar@{->}[rr] \ar@{->}[rruu]
&& \text{$K(1,4)$} \ar@{->}[rru] 
&&
&&
}
$$
\label{fig:confluence}
\end{figure}

\noindent
Here $K(1,1,1,1,1)$ designates the $2$-dimensional Garnier system and, 
in general, the symbol $(\#)=(r_1, \ldots, r_m)$ means that its underlying 
monodromy preserving deformation is concerned with a linear equation 
with $m$ singular points of Poincar\'e ranks $r_1-1, \ldots, r_m-1$. 
In what follows a hypergeometric differential system with two independent 
variables is called a hypergeometric system of type $(\#)$ when it is 
associated with a confluent $2$-dimensional Garnier system $K(\#)$. 
Among them, in this article, we consider the hypergeometric systems of 
type $(1,4)$ and $(2,3)$, or the following two third-order ordinary 
differential equations obtained from the hypergeometric systems of 
type $(1,4)$ and $(2,3)$ by fixing the second variable $x_2 = t$: 
The first one is 
\begin{equation}
\label{eq:intro:(1,4)_eq(d/dx)} 
	\left\{ 3 \hbar^3 \frac{d^3}{dx^3} + 2t \hbar^2 \frac{d^2}{dx^2} + x \hbar \frac{d}{dx} 
			- \hat{\lambda}_\infty \right\} \psi = 0, 
\end{equation} 
which is called the hypergeometric equation of type $(1,4)$, and the second one is 
\begin{equation}
\label{eq:intro:(2,3)_eq(d/dx)}
	\left\{ 4 \hbar^3 \frac{d^3}{dx^3} - 2 x \hbar^2 \frac{d^2}{dx^2} 
			+ 2 ( \hat{\lambda}_\infty - \hbar ) \hbar \frac{d}{dx} - t \right\} \psi = 0, 
\end{equation}
which is called the hypergeometric equation of type $(2,3)$. 

The purpose of this article is to show that the Voros coefficients 
of \eqref{eq:intro:(1,4)_eq(d/dx)} and \eqref{eq:intro:(2,3)_eq(d/dx)} 
are expressed as the difference values of the free energy 
defined through the topological recursion due to Eynard and Orantin \cite{EO} 
(\thmref{main(i)}) and further that the explicit forms 
of Voros coefficients and the free energy of 
\eqref{eq:intro:(1,4)_eq(d/dx)} and \eqref{eq:intro:(2,3)_eq(d/dx)} can be 
obtained by using this relation between these two quantities 
(\thmref{main(iii)} and \thmref{main(iv)}). 

Voros coefficients are defined as contour integrals of the logarithmic 
derivative of WKB solutions in the thoery of the exact WKB analysis. 
Its importance in the study of the global behavior of solutions has been already recognized by the pioneering work of Voros (\cite{Voros83}). The explicit form of Voros coefficients plays an important role to describe parametric Stokes phenomena, which are Stokes phenomena with respect to parameters included in the equation. 
The explicit form of the Voros coefficients is now known for the confluent family of the Gauss hypergeometric equation and the hypergeometric equation of type (1,4) (\cite{SS, Takei08, KoT11, ATT, Aoki-Tanda, AIT, IKo}). 
Voros coefficients are also studied for the Painlev\'e equations to clarify the parametric Stokes phenomena (\cite{I14}). 

On the other hand, the topological recursion introduced by Eynard and 
Orantin (\cite{EO}) is a generalization of the loop equations that the 
correlation functions of the matrix model satisfy. For a Riemann surface 
$\Sigma$ and meromorphic functions $x$ and $y$ on $\Sigma$, it 
produces an infinite tower of meromorphic differentials $W_{g,n}(z_1, \ldots, z_n)$ on $\Sigma$. 
A triplet $(\Sigma, x, y)$ is called a spectral curve and $W_{g,n}(z_1, \ldots, z_n)$ is called a correlation function. As is shown in \cite{GS, DM, BE} etc., 
the quantization scheme connects WKB solutions of 
differential equations with the topological recursion. More precisely, 
WKB solutions of a differential equation are constructed also by 
correlation functions for the spectral curve corresponding to the 
classical limit of the differential equation provided that the spectral curve 
satisfies the so-called ``admissibility condition" (cf. \ \cite[Definition 2.7]{BE}). 
Moreover, for a spectral curve, we can define free energies 
(also called symplectic invariants) $F_g$. For more details about the topological recursion, 
see, e.g., the review paper \cite{EO-08}. 
The main results of this paper as well as those of \cite{IKT-part1,IKT-part2} 
strengthen the interplay between the WKB theory and the topological recursion 
and these interplays are expected to produce more profound insights in these theories. 

%
The paper is organized as follows: 
In \S2 we recall some fundamental facts about 
the exact WKB analysis and Eynard-Orantin's topological recursion. 
In \S3 we study quantization of the spectral curve. 
In \S4 we state our main theorem (\thmref{main(i)}, \thmref{main(ii)}, \thmref{main(iii)} and \thmref{main(iv)}). 
We give a proof of our result only for (1,4) curve, but (2,3) curve can be treated similarly. 
%
\section*{Acknowledgement}
The author would like to express my special thanks to late Professor Tatsuya Koike. 
The author is also very grateful to Professors 
Takashi Aoki, 
Sampei Hirose, 
Kohei Iwaki, 
Shingo Kamimoto, 
Takahiro Kawai, 
Saiei Matsubara, 
Genki Shibukawa, 
Takahiro Shigaki, 
Nobuki Takayama, 
Yoshitsugu Takei 
and 
Mika Tanda 
for helpful discussions and communications. 
%



\section{Voros coefficients and the topological recursion}
\label{sec:review}


\subsection{WKB solutions}
\label{subsec:WKB-sol}

In this article we discuss the third order ordinary differential equation with a small parameter $\hbar \ne 0$ 
of the form
\begin{equation}
\label{eq:3rd-ODE}
	\left\{
		p_0(x, \hbar) \hbar^3 \frac{d^3}{dx^3} + p_1(x, \hbar) \hbar^2 \frac{d^2}{dx^2} 
		+ p_2(x, \hbar) \hbar \frac{d}{dx} + p_3(x, \hbar)
	\right\}\psi = 0,
\end{equation}
where $x \in \mathbb{C}$, and
\begin{equation}
	p_i(x, \hbar) = p_{i,0}(x) + \hbar p_{i,1}(x) \quad (i = 0, 1, 2, 3)
\end{equation}
with rational functions $p_{i,j}(x)$ ($i = 0, 1, 2, 3$, $j = 0, 1$) and 
\begin{equation}
	p_{0,1}(x) = p_{1,1}(x) = 0. 
\end{equation} 
We consider \eqref{eq:3rd-ODE} as a differential equations on the Riemann sphere $\mathbb{P}^1$ 
with regular or irregular singular points. 
For \eqref{eq:3rd-ODE} we can construct a formal solution, called a WKB solution, of the form 
\begin{align}
\label{eq:WKB-type}
	\psi (x, \hbar) 
		&= \exp \left[ \int^x S(x, \hbar) dx \right]. 
\end{align}
The logarithmic derivative $S(x, \hbar)$ of WKB solutions of \eqref{eq:3rd-ODE} satisfies the equation 
\begin{equation}
\begin{split}
\label{eq:Riccati-gen}
	p_0(x, \hbar) \hbar^3 
		\left( \frac{d^2}{dx^2} S(x, \hbar) + 3 S(x, \hbar) \frac{d}{dx} S(x, \hbar) + {S(x, \hbar)}^3 \right) 
	+ p_1(x, \hbar) \hbar^2 \left( \frac{d}{dx} S(x, \hbar) + {S(x, \hbar)}^2 \right) \\
	\qquad + \hbar p_2(x, \hbar) S(x, \hbar) + p_3(x, \hbar) = 0. 
\end{split}
\end{equation}
Eq. \eqref{eq:Riccati-gen} is a counterpart of the Riccati equation in the second-order case 
and admits a solution of the form 
\begin{align}
\label{eq:Riccati-gen-expansion}
	S(x, \hbar) 
		&:= \hbar^{-1} S_{-1}(x) + S_0(x) + \hbar S_1(x) + \cdots 
		= \sum_{m = -1}^{\infty} \hbar^m S_m(x). 
\end{align}
In fact, by substituting \eqref{eq:Riccati-gen-expansion} into \eqref{eq:Riccati-gen}, and comparing like powers of both sides  with respect to $\hbar$, we obtain 
\begin{align}
\label{eq:Riccati-gen-1}
	p_{0,0}(x) S_{-1}^3 + p_{1,0}(x) S_{-1}^2 + p_{2,0}(x) S_{-1} + p_{3,0}(x)  &= 0,\\
\label{eq:Riccati-gen-2}
	\left(3 p_{0,0}(x) S_{-1}{}^2 +2 p_{1,0}(x) S_{-1} +p_{2,0}(x)\right) S_0 
	+ 3 p_{0,0}(x) S_{-1} \frac{d S_{-1}}{dx} +p_{1,0}(x) \frac{d S_{-1}}{dx} \\
	+p_{2,1}(x) S_{-1} +p_{3,1}(x) &= 0, \notag 
\end{align}
and
\begin{align}
\label{eq:Riccati-gen-3}
	\left(3 p_{0,0}(x) S_{-1}{}^2 +2 p_{1,0}(x) S_{-1} +p_{2,0}(x)\right) S_{m + 1} 
	+ \sum_{\substack{i + j + k = m-1 \\ i, j, k \geq  0}} S_{i} S_{j} S_{k} + 3 \sum_{j = 0}^{m-1} S_{m - j - 1} S_{j} 
	\\ 
	+ 3 p_{0,0}(x) S_m \frac{d S_{-1}}{dx} + 3 p_{0,0}(x) S_{-1} \frac{d S_{m}}{dx} 
	+ p_{0,0}(x)\frac{d^2 S_{m-1}}{dx^2} 
	+ p_{1,0}(x) \sum_{j = 0}^m S_{m - j} S_{j} \notag \\
	+ p_{1,0}(x) \frac{d S_{m}}{dx} + p_{2,1}(x) S_m = 0 \quad (m \geq 0). \notag 
\end{align} 
Eq. \eqref{eq:Riccati-gen-1} has three solutions, and once we fix one of them, 
we can determine $S_m$ for $m \geq 0$ uniquely and recursively 
by \eqref{eq:Riccati-gen-2} and \eqref{eq:Riccati-gen-3}.


\subsection{Voros coefficients}
\label{subsec:Voros-coeff}

A Voros coefficient is defined as a properly regularized integral of $S(x, \hbar)$ along 
a path connecting singular points of \eqref{eq:3rd-ODE}. 
When $S_m(x)$ with $m \geq 1$ is integrable at any singular point of \eqref{eq:3rd-ODE}, 
we can define Voros coefficients by 
\begin{equation}\label{eq:def-Voros-coeff}
V_{\gamma_{b_1, b_2}}(\hbar)
:= \int_{\gamma_{b_1, b_2}}
\big( S(x, \hbar) -\hbar^{-1}S_{-1}(x) - S_0(x)\big) dx
= \sum_{m = 1}^{\infty} \hbar^m  \int_{\gamma_{b_1, b_2}} S_m(x) dx,
\end{equation}
where $\gamma_{b_1, b_2}$ is a path from a singular point $b_1$ to a singular point $b_2$.
(When there is no need to specify a path $\gamma_{b_1,b_2}$, we use the abbreviated 
notation $V(\hbar)$ instead of $V_{\gamma_{b_1, b_2}}(\hbar)$.)
Note that Voros coefficients only depend on the class 
$[\gamma_{b_1, b_2}]$ of paths in the relative  homology group
$$
H_1 \big(\mathbb{P}^1 \setminus
\{\text{Turning points}\},
\{\text{Singular points}\}; \mathbb{Z} \big).
$$
Such an integration contour (or a relative homology class) 
can be understood as a lift of a path on $x$-plane 
onto the Riemann surface of $S_{-1}(x)$ 
(i.e., three sheeted covering of $x$-plane). 
The lift of a path is specified by drawing branch cuts and distinguishing the 
first, second and third sheets of the Riemann surface.


\subsection{The global topological recursion}
\label{sec:TR}

Let us first fix notation. 
We restrict ourselves to the case when a spectral curve is of genus $0$ 
because we will not discuss the general case in this paper 
(see \cite{BE-12} for the general definition). 

\begin{dfn} \label{def:spectral-curve}
A spectral curve (of genus $0$) is a pair $(x(z), y(z))$
of non-constant rational functions on $\mathbb{P}^1$, 
such that their exterior differentials 
$dx$ and $dy$ never vanish simultaneously. 
\end{dfn}

Let $R$ be the set of ramification points of $x(z)$, 
i.e., $R$ consists of zeros of $dx(z)$ of any order and poles of $x(z)$ 
whose orders are greater than or equal to two 
(here we consider $x$ as a branched covering map from $\mathbb{P}^1$ to itself). 
We further assume that 

\begin{itemize}
\item[(A1)]
A function field $\mathbb{C}(x(z), y(z))$ coincides with $\mathbb{C}(z)$.

\item[(A2)]

If $r$ is a ramification point which is a pole of $x(z)$, 
and if $Y(z) = - x(z)^2 y(z)$ is holomorphic near $r$,
then $dY(r) \neq 0$.

\item[(A3)]
All of the ramification points of $x(z)$ are simple,
i.e., the ramification index of each ramification point
is two.

\item[(A4)]
We assume branch points are all distinct,
where a branch point is defined as the image of
a ramification point by $x(z)$.
\end{itemize}

We need to introduce some notation to define the topological recursion. 

\begin{dfn} \label{def:effective-ramification}
A ramification point $r$ is said to be ineffective if 
the correlation functions $W_{g,n}(z_1,\dots,z_n)$ 
for $(g,n) \ne (0,1)$ are holomorphic at $z_i = r$ for each $i=1,\dots,n$. 
A ramification point which is not ineffective is called effective. 
The set of effective ramification points is denoted by $R^{\ast}$ $(\subset R)$. 
\end{dfn}


\begin{dfn}
For two sets $A$ and $B$, $A \subseteq_k B$ means $A \subseteq B$ and $|A| = k$. 
\end{dfn}

\begin{dfn}
$\mathcal{S}(\bm{t})$ denotes the set of set partitions of $\bm{t} = \{ t_1, \ldots, t_k \}$. 
\end{dfn}

Then, we define the recursive structure: 

\begin{dfn}[{\cite[Definition 3.4]{BE}}]
Let $\{ W_{g, n} \}$ be an arbitrary collection of symmetric multidifferential on $(\mathbb{P}^1)^n$ 
with $g \geq 0$ and $n \geq 1$. Let $k \geq 1$, 
$\bm{t} = \{ t_1, \ldots, t_k \}$ and $\bm{z} = \{ z_1, \ldots, z_n \}$. Then, we define 
\begin{align}
\label{eq:R(k)}
	{\mathcal{R}}^{(k)} \left(W_{g, n+1}(\bm{t}; \bm{z})\right) 
		&:= \sum_{\mu \in \mathcal{S}(\bm{t})} 
			\sum_{\sqcup_{i=1}^{l(\mu)} I_i = \{1, 2, \cdots, n\}} 
			\sum'_{\sum_{i=1}^{l(\mu)} g_i = g + l(\mu) - k} 
			\left\{ \prod_{i=1}^{l(\mu)} W_{g_i, |\mu_i| + |I_i|}(\mu_i, z_{I_i}) \right\}. 
\end{align}
The first summation in \eqref{eq:R(k)} is over set partitions of $\bm{t}$, 
$l(\mu)$ is the number of subsets in the set partition $\mu$. 
The third summation in \eqref{eq:R(k)} is over all $l(\mu)$-tuple 
of non-negative integers $(g_1, \ldots, g_{l(\mu)})$ such that $\sum_{i=1}^{l(\mu)} g_i = g + l(\mu) - k$.  
$\sqcup$ denotes the disjoint union, 
and the prime ${}'$ on the summation symbol in \eqref{eq:R(k)} means that we exclude terms for
$(g_i, |\mu_i| + |I_i|) = (0, 1)$ ($i = 1, \ldots, l(\mu)$)
(so that $W_{0, 1}$ does not appear) in the sum. 
We also define 
\begin{align}
	{\mathcal{R}}^{(0)} W_{g, n+1}(\bm{z}) &:= \delta_{g,0} \delta_{n,0},  
\end{align}
where $\delta_{i,j}$ is the Kronecker delta symbol. 
\end{dfn}

\begin{ex} 
For $k=2$, $\mathcal{S}(\bm{t})$ is given by 
\begin{align}
\mathcal{S}(\{ t_1, t_2 \}) 
= \Bigl\{ \bigl\{ \{ t_1, t_2 \} \bigr\}, \bigl\{\{ t_1 \}, \{ t_2 \} \bigr\} \Bigr\}. 
\end{align}
Therefore, we have 
\begin{align}
{\mathcal{R}}^{(2)} \left(W_{g, n+1}(\bm{t}; \bm{z})\right) 
&= W_{g-1,n+2}(\bm{t}, \bm{z}) 
	+ \sum_{I_1\sqcup I_2 = \{1, 2, \cdots, n\} } 
		\sum'_{g_1 + g_2 = g - 1} 
		\left\{ \prod_{i=1}^{2} W_{g_i, 1 + |I_i|}(t_i, z_{I_i}) \right\}. 
\end{align}
\end{ex}
 
We now define the topological recursion. 

\begin{dfn}[{\cite[Definition 3.6]{BE}}]
Eynard-Orantin's correlation function
$W_{g, n}(z_1, \cdots, z_n)$ for $g \geq 0$ and $n \geq 1$ 
is defined as a multidifferential 
on $(\mathbb{P}^1)^n$ using the recurrence relation
\begin{align}
\label{eq:gTR}
	&W_{g, n+1}(z_0, z_1, \cdots, z_n) \\
	&:= \sum_{r \in R} \Res_{z = r} \left\{ 
			\sum_{k=1}^{r-1} \sum_{\beta(z) \subseteq_k \tau'(z)} 
			(-1)^{k+1} \frac{w^{z - \alpha}(z_0)}{E^{(k)}(z;\beta(z))} 
			{\mathcal{R}}^{(k+1)} \left(W_{g, n+1}(z, \beta(z);z_1, \cdots, z_n)\right) 
		\right\} \notag
\end{align}
for $2g + n \geq 2$ with initial conditions
\begin{align}
W_{0, 1}(z_0) &:= y(z_0) dx(z_0),
\quad
W_{0, 2}(z_0, z_1) = B(z_0, z_1)
:= \frac{dz_0 dz_1}{(z_0 - z_1)^2}.
\end{align}
Here we set $W_{g,n} \equiv 0$ for a negative $g$ and 
\begin{align}
\label{eq:E(k)}
	E^{(k)}(z; t_1, \ldots, t_k)
		&:= \prod_{i=1}^k (W_{0,1}(z) - W_{0,1}(t_i)). 
\end{align}
The second and third summations in \eqref{eq:gTR} together mean that 
we are summing over all subsets of $\tau'(z)$. 
$\alpha$ is an arbitrary base point on $\mathbb{P}^1$,  but it can be checked (see \cite{BE-12}) that the definition is actually independent of the choice of base point $\alpha$.
We have also used the multi-index notation:
for $I = \{i_1, \cdots, i_m\} \subset \{1, 2, \cdots, n\}$
with $i_1 < i_2 < \cdots < i_m$, $z_I:= (z_{i_1}, \cdots, z_{i_m})$.
\end{dfn}

Note that this recursion was called ``global topological recursion" in \cite{BE-12}. 
It was shown in \cite{BE-12} that it is indeed equivalent to the following 
usual local formulation of the topological recursion when the ramification points are all simple. 


\subsection{Free energy through the topological recursion}
\label{subsec:free-energy}

The $g$-th free energy $F_g$ ($g\geq 0$) is a complex number
defined for the spectral curve,
and one of the most important objects in Eynard-Orantin's theory.
It is also called a symplectic invariant since it is 
``almost'' invariant under symplectic transformations
of spectral curves (see \cite{EO-13} for the details). 

\begin{dfn}[{\cite[Definition 4.3]{EO}}]
For $g \geq 2$, the $g$-th free energy $F_g$ is defined by
\begin{equation}
\label{def:Fg2}
F_g := \frac{1}{2- 2g} \sum_{r \in R} \Res_{z = r}
\big[\Phi(z) W_{g, 1}(z) \big]
\quad (g \geq 2),
\end{equation}
where $\Phi(z)$ is a primitive of $y(z) dx(z)$. 
For $g=1$, we define the free energy $F_1$ satisfying \eqref{eq:variational_free-energy}. 
The free energies $F_0$ for $g=0$ is also defined, but in a different manner 
(see \cite[\S 4.2.3]{EO} for the definition). 
\end{dfn}
Note that the right-hand side of \eqref{def:Fg2} does not
depend on the choice of the primitive
because $W_{g, 1}$ has no residue at each ramification point. 

In applications (and in our article), the generating series
\begin{equation} \label{eq:total-free-energy}
F := \sum_{g = 0}^{\infty} \hbar^{2g-2} F_g
\end{equation}
of $F_g$'s is crucially important. 
We also call the generating series \eqref{eq:total-free-energy} 
the free energy of the spectral curve. 


\subsection{Variational formulas for the correlation functions}
\label{subsec:variational-formula-TR}

In \S \ref{subsection:quantum-(1,4)} and \S \ref{subsection:quantum-(2,3)} 
we will consider a family of spectral curves parametrized by complex parameters. 
For our purpose, we briefly recall the variational formulas obtained 
by \cite[\S 5]{EO} which describe the differentiation 
of the correlation functions $W_{g,n}$ and the free energies $F_g$ 
with respect to the parameters.

Suppose that we have given a family 
$(x_\varepsilon(z), y_\varepsilon(z))$
of spectral curves parametrized by a complex parameter 
$\varepsilon$ which lies on a certain domain $U \subset {\mathbb C}$ 
such that  
\begin{itemize}
\item 
$x_\varepsilon(z), y_\varepsilon(z)$ depend 
holomorphically on $\varepsilon \in U$. 
\item 
$x_\varepsilon(z), y_\varepsilon(z)$ 
satisfy the assumptions (A1) -- (A4) 
for any $\varepsilon \in U$. 
\item 
The cardinality of the set $R_\varepsilon$ of 
ramification points of $x_\varepsilon(z)$ is constant on $\varepsilon \in U$
(i.e. ramification points of $x_\varepsilon(z)$ 
are distinct for any $\varepsilon \in U$).
\end{itemize}
Then, the correlation functions $W_{g,n}(z_1, \dots, z_n; \varepsilon)$ 
and the $g$-th free energy $F_g(\varepsilon)$ defined from the spectral curve 
$(x_\varepsilon(z), y_\varepsilon(z))$
are holomorphic in $\varepsilon \in U$ 
as long as $z_i \notin R_\varepsilon$ for any $i=1,\dots,n$. 

In order to formulate a variational formula for correlation functions, 
we need to introduce the notion of ``differentiation with fixed $x$". 
For a meromorphic differential $\omega(z; \varepsilon)$ on ${\mathbb P}^1$, 
which depends on $\varepsilon$ holomorphically, define 
\begin{equation}
\delta_{\varepsilon} \, \omega(z; \varepsilon)  
:= \left( 
\frac{\partial}{\partial \varepsilon} \omega(z_{\varepsilon}(x); \varepsilon) 
\right) \biggl|_{x=x_{\varepsilon}(z)} 
\quad (z \notin R_\varepsilon), 
\end{equation}
where $z_{\varepsilon}(x)$ is (any branch of) the inverse function of 
$x = x_{\varepsilon}(z)$ which is defined away from branch points 
(i.e. points in $x_{\varepsilon}(R_\varepsilon)$). 
In \cite{EO} the notation 
$\delta_{\Omega} \, \omega(z; \varepsilon) \big|_{x(z)}$ 
is used for $\delta_{\varepsilon} \omega(z; \varepsilon)$ defined above.
Such differentiation $\delta_{\varepsilon}$ can be generalized 
to multidifferentials in an obvious way. 
Then, under these assumptions, the variational formula is formulated as follows.

\begin{thm}[{\cite[Theorem 5.1]{EO}}] \label{thm:VariationFormula}
In addition to the above conditions, for any $\varepsilon \in U$, 
we further assume that  
\begin{itemize}
\item 
If $r_\varepsilon \in R_\varepsilon$ is a zero of 
$dx_\varepsilon(z)$, then the functions
$\partial x_\varepsilon/ \partial \varepsilon$ and 
$\partial y_\varepsilon/ \partial \varepsilon$ are holomorphic
(as functions of $z$) at $r_\varepsilon$, 
and $dy_\varepsilon(z)$ does not vanish
(as a differential of $z$) at $r_\varepsilon$.
\item 
If $r_\varepsilon \in R_\varepsilon$ is a pole of $x_\varepsilon(z)$ 
with an order greater than or equal to two, then 
\[
\frac{\Omega_\varepsilon(z) \, B(z_1, z) \, B(z_2 , z)}
{dy_\varepsilon(z) dx_\varepsilon(z)}
\]
is holomorphic (as a differential in $z$) at $r(\varepsilon)$, where 
\begin{equation} \label{eq:Omega}
\Omega_\varepsilon(z) := 
\frac{\partial y_\varepsilon}{\partial \varepsilon}(z) \, dx(z)
- \frac{\partial  x_\varepsilon}{\partial \varepsilon}(z) \, dy(z).
\end{equation}
\item 
There exist a path $\gamma$ in $\mathbb{P}^1$ passing through
no ramification point and a function $\Lambda_\varepsilon (z)$ 
holomorphic in a neighborhood of $\gamma$ for which the following holds.
\begin{equation}
\Omega_\varepsilon(z) = 
\int_{\zeta \in \gamma} \Lambda_\varepsilon(\zeta) \, B(z, \zeta).
\end{equation} 
\end{itemize}
Then, $W_{g,n}(z_1, \dots, z_n; \varepsilon)$ 
and $F_g(\varepsilon)$ defined from the spectral curve 
$(x_\varepsilon(z), y_\varepsilon(z))$
satisfy the following relations:  
\begin{itemize}
\item[{\rm{(i)}}]
For $2g + n \geq 2$, 
\begin{equation}
\delta_{\varepsilon} \, W_{g, n} 
(z_1, \cdots, z_n; \varepsilon)
= \int_{\zeta \in \gamma} \Lambda_\varepsilon(\zeta) \,
W_{g, n + 1}(z_1, \cdots, z_n, \zeta; \varepsilon) 
\end{equation}
holds on $\varepsilon \in U$ as long as each of $z_1, \cdots, z_n$ satisfies 
$z_i \notin R_\varepsilon$.

\item[{\rm{(ii)}}]
For $g \geq 1$,
\begin{equation}
\label{eq:variational_free-energy}
\frac{\partial F_g}{\partial \varepsilon}(\varepsilon)
= \int_{\gamma}\Lambda_\varepsilon(z) \, W_{g, 1}(z;\varepsilon)
\end{equation}
holds on $\varepsilon \in U$.


\end{itemize}

\end{thm}

See \cite[\S 5.1]{EO} 
(based on the Rauch's variation formula; see \cite{KK} for example) 
for the proof.
We note that, since we modify the definition of the topological recursion 
by adding higher order poles of $x(z)$ as ramification point, 
we also need to require the second condition in the above claim. 



\section{Quantization of spectral curves}
\label{subsec:quantum-curve}


We treat the quantization by using the divisor with parameters which was introduced by \cite{BE}. 
 In this article, we consider the defining equation of the spectral curve
\begin{equation}
\label{eq:spectral-curve}
	P(x, y) = p_0(x) y^3 + p_1(x) y^2 + p_2(x) y + p_3(x) = 0. 
\end{equation}

\begin{dfn}[{\cite[Definition 2.3]{BE}}]
Let us rewrite the defining equation \eqref{eq:spectral-curve} of the spectral curve as 
\begin{equation}
	P(x, y) = \sum_{i, j \in A} \alpha_{i, j} x^i y^j = 0 \quad (\alpha_{i, j} \ne 0).
\end{equation}
Then the Newton polygon $\Delta$ of \eqref{eq:spectral-curve} is the convex hull of the set $A$. 
\end{dfn}

\begin{dfn}[{\cite[Definition 2.5]{BE}}]
For $m = 2, 3$, we define the following meromorphic function on $\mathbb{P}^1$:  
\begin{equation}
	P_m(x, y) = \sum_{k = 1}^{m-1} p_{m-1-k}(x) y^k = 0. 
\end{equation}
\end{dfn}

\begin{dfn}[{\cite[Definition 2.7]{BE}}]
We say that a spectral curve is admissible if: 
\begin{itemize}
	\item[1.] Its Newton polygon $\Delta$ has no interior point; 
	\item[2.] If the origin $(x, y) = (0, 0) \in \mathbb{C}^2$ is on the curve 
			$\{ P(x, y) = 0 \subset \mathbb{C}^2\}$, then the curve is smooth at this point. 
\end{itemize}
\end{dfn}

We assume that our spectral curve $(x(z), y(z))$ is admissible. 
Then the following theorem holds according to \cite{BE}. 

\begin{thm}[{\cite[Lemma 5.14]{BE}}]
\label{thm:WKB-Wg,n-BE}
Let $\beta_i \quad (1 \leqq i \leqq n)$ be simple poles of $x(z)$ and 
\begin{equation}
\label{eq:D}
\begin{split}
	D(z ; \underline{\nu}) &= [z] - \sum_{i=1}^{n} \nu_i [\beta_i]
\end{split}
\end{equation}
be a divisor on $\mathbb{P}^1$, where $\nu_i \quad (1 \leqq i \leqq n)$ are complex numbers satisfying 
$\sum_{i=1}^{n} \nu_i = 1$. 
For a differential $\omega(z)$, we define its integration along the divisor $D(z; \underline{\nu})$ by   
\[
\int_{D(z; \underline{\nu})} \omega(z) = 
\sum_{i=1}^{n} \nu_i \int^{z}_{\beta_i} \omega(z) 
\]
and extend the definition to multidifferentials in an obvious way. 
Let $W_{g, n}(z_1, \cdots, z_{n})$ be the correlation functions of a spectral curve $(x(z), y(z))$ defined from (\ref{eq:spectral-curve}). 
Then,
\begin{equation}
\label{eq:WKB-Wg,n}
\begin{split}
	\psi(x, \hbar)
	&= \exp \Bigg[ \hbar^{-1} \int^z W_{0, 1}(z) 
		+ \frac{1}{2!} \int_{D(z ; \nu)} \int_{D(z ; \nu)} 
			\left( W_{0, 2}(z_1, z_2) - \frac{dx(z_1) \, dx(z_2)}{(x(z_1) - x(z_2))^2} \right) \\
	&\quad \left. \left.
	   + \sum_{m = 1}^{\infty} \hbar^m 
	  	\left\{ \sum_{\substack{2g + n - 2 = m \\ g \geq 0, \, n \geq 1}} 
			\frac{1}{n!} \int_{D(z ; \nu)} \cdots \int_{D(z ; \nu)} W_{g, n}(z_1, \ldots, z_n) 
		\right\} \right]  \right|_{z = z(x)}
\end{split}
\end{equation}
is a WKB type formal solution of 
\begin{equation}
\begin{split}
\label{eq:quantization}
	\left[ D_1 D_2 \frac{p_0(x)}{x^{\lfloor \alpha_{3} \rfloor}} D_{3} 
			+ D_1 \frac{p_1(x)}{x^{\lfloor \alpha_{2} \rfloor}} D_{2} 
			+ \frac{p_2(x)}{x^{\lfloor \alpha_{1} \rfloor}} D_{1} 
			+ \frac{p_3(x)}{x^{\lfloor \alpha_{0} \rfloor}} 
			- \hbar C_1 D_1 \frac{x^{\lfloor \alpha_{2} \rfloor}}{x^{\lfloor \alpha_{1} \rfloor}} 
			- \hbar C_2 \frac{x^{\lfloor \alpha_{1} \rfloor}}{x^{\lfloor \alpha_{0} \rfloor}} 
	\right] \psi = 0, 
\end{split}
\end{equation}
where 
\begin{align*}
	\alpha_m &= \inf \{ a \mid (a, m) \in \Delta \} \quad (m = 0, 1, 2, 3), \\
	D_i 
	&= \hbar \frac{x^{\lfloor \alpha_{i} \rfloor}}{x^{\lfloor \alpha_{i-1} \rfloor}} \frac{d}{dx} 
		\quad (i = 1, 2, 3), \\
	C_k 
	&= \sum_{i = 1}^{n} \nu_i \left( 
		\lim_{z \rightarrow {\beta_i}} \frac{P_{k+1}(x(z), y(z))}{{x(z)}^{\lfloor \alpha_{3-k} \rfloor + 1}}
		\right) \quad (k = 1, 2). 
\end{align*}
\end{thm}

\begin{rem}
It is mentioned by \cite[Remark 5.12]{BE} that 
it is also possible to choose a pole of $x$ of order more than one as $\beta_i$ in \eqref{eq:D}
when $\beta \notin R^{\ast}$. 
Therefore, we can use \thmref{WKB-Wg,n-BE} in the case (1,4) and (2,3) curve in the next section. 
\end{rem}


\subsection{Quantum (1,4) curve}
\label{subsection:quantum-(1,4)}

Let us consider the (1,4) curve defined by 
\begin{equation}
\label{eq:(1,4)_P(x,y)}
	P(x, y) = 3 y^3 + 2t y^2 + x y - {\lambda_\infty} = 0, 
\end{equation}
with parameters $t, {\lambda_\infty} \ne 0$. 
A rational parameterization of this curve is 
\begin{equation}
\label{eq:(1,4)_parameterization}
\begin{cases}
	\displaystyle
	x = x(z) 
	= \frac{-3 z^3 - 2t z^2 + {\lambda_\infty}}{z} = -3 z^2 - 2t z + \frac{{\lambda_\infty}}{z}, \\[10pt]
	\displaystyle
	y = y(z) = z. 
\end{cases}
\end{equation}
First few terms of the correlation functions and free energies are computed as 
\begin{align*}
	W_{0, 3}(z_1, z_2, z_3) 
	&= \biggl\{ 
		\frac{2 z_1(15 {z_1}^5 - (9{z_2} + 9{z_3} - 4t) {z_1}^4 
				- (2t{z_2} + 2t{z_3} - 3{z_2}{z_3}) {z_1}^3 + {\lambda_\infty} {z_1}^2 
				- {\lambda_\infty} {z_2}{z_3})}
			{({z_1} - {z_2})^3 ({z_1} - {z_3})^3 (6 {z_1}^3 + 2t {z_1}^2 + {\lambda_\infty})^2} \\ 
	&\quad
		+ \frac{2 z_2(15 {z_2}^5 - (9{z_3} + 9{z_1} - 4t) {z_2}^4 
				- (2t{z_3} + 2t{z_1} - 3{z_3}{z_1}) {z_2}^3 + {\lambda_\infty} {z_2}^2 
				- {\lambda_\infty} {z_3}{z_1})}
			{({z_2} - {z_3})^3 ({z_2} - {z_1})^3 (6 {z_2}^3 + 2t {z_2}^2 + {\lambda_\infty})^2} \\ 
	&\quad
		+ \frac{2 z_3(15 {z_3}^5 - (9{z_1} + 9{z_2} - 4t) {z_3}^4 
				- (2t{z_1} + 2t{z_2} - 3{z_1}{z_2}) {z_3}^3 + {\lambda_\infty} {z_3}^2 
				- {\lambda_\infty} {z_1}{z_2})}
			{({z_3} - {z_1})^3 ({z_3} - {z_2})^3 (6 {z_3}^3 + 2t {z_3}^2 + {\lambda_\infty})^2}
		\biggl\} \\
	&\quad
		\times d{z_1} \, d{z_2} \, d{z_3}, \\
	W_{1, 1}(z)
	&= \frac{z^2 (27 z^6 - 99 z^3 - 36 {\lambda_\infty} t z^2 - 4 {\lambda_\infty} t^2 z + 3 {\lambda_\infty}^2)}
			{(6 z^3 + 2t z^2 + {\lambda_\infty})^4} \, dz,
\end{align*}
\begin{align*}	
	F_0(\lambda_\infty, t)
	&= - \frac{t^6}{972} + \frac{2 {\lambda_\infty} t^3}{27} - \frac{3 {\lambda_\infty}^2}{4} 
		+ \frac{{\lambda_\infty}^2}{4} \log{(-3 {\lambda_\infty}^2)}, \quad
	F_1(\lambda_\infty, t) = - \frac{1}{12} \log{\lambda_\infty}.
\end{align*}
\begin{rem}
	It seems $W_{0,3}$ has singularities at  $z_1 = z_2 = z_3$, 
	but we can verify that $W_{0,3}$ is holomorphic there. 
\end{rem}

We choose 
\begin{equation}
\label{eq:(1,4)_D}
\begin{split}
	D(z ; \nu) 
	&= [z] - (1 - \nu_\infty) [0] - \nu_\infty [\infty] \\
	&= (1 - \nu_\infty) ([z] - [0]) + \nu_\infty ([z] - [\infty]) 
\end{split}
\end{equation}
as the divisor for the quantization. 
\begin{rem}
	$z = \infty$ is a double pole of $x(z)$, i.e., $\infty \in R$, but we can verify that $\infty \notin R^{\ast}$. 
	Therefore, we can choose $\beta = \infty$ as a base point. 
\end{rem}
Then, \thmref{WKB-Wg,n-BE} gives the quantum curve of the (1,4) curve (quantum (1,4) curve): 
\begin{equation}
\label{eq:(1,4)_eq(d/dx)} 
	\left\{ 3 \hbar^3 \frac{d^3}{dx^3} + 2t \hbar^2 \frac{d^2}{dx^2} + x \hbar \frac{d}{dx} 
			- \hat{\lambda}_\infty \right\} \psi = 0. 
\end{equation} 
Here we used the notation 
\begin{equation} \label{eq:lambda-hat-(1,4)}
	\hat{\lambda}_\infty = \lambda_\infty - \nu_{\infty} \hbar.
\end{equation}
Note that the special case $t = 0$
of the equation has been already constructed as 
a quantum curve in \cite[\S6.2.2]{BE}. 

Let $S_m(x, \lambda, \nu)$ be the coefficient of the Voros coefficient of \eqref{eq:(1,4)_eq(d/dx)}. Then $S_m(x, \lambda, \nu)$ satisfies the following lemma. 
\begin{lem}
\label{lem:(1,4)_Sn} 
For $m = 1, 2, \cdots$, we have
\begin{align} 
S_m(x, \lambda, \nu) = O(x^{-{2}}) 
\quad
(x \rightarrow \infty).
\end{align} 
\end{lem} 


\subsection{Quantum (2,3) curve}
\label{subsection:quantum-(2,3)}

Let us consider the (2,3) curve defined by 
\begin{equation}
\label{eq:(2,3)_P(x,y)}
	P(x, y) = 4 y^3 - 2 x y^2 + 2 {\lambda_\infty} y - t = 0
\end{equation}
with parameters $t, {\lambda_\infty} \ne 0$. 
A rational parameterization of this curve is 
\begin{equation}
\label{eq:(2,3)_parameterization}
\begin{cases}
	\displaystyle
	x = x(z) 
	= \frac{4 z^3 + 2 {\lambda_\infty} z - t}{2 z^2} = 2 z + \frac{{\lambda_\infty}}{z} - \frac{t}{2 z^2} \\[10pt]
	\displaystyle
	y = y(z) = z. 
\end{cases}
\end{equation}
First few terms of the correlation functions and free energies are computed as 
\begin{align*} 
	W_{0, 3}(z_1, z_2, z_3) 
	&= \biggl\{ 
		- \frac{{z_1}^2 (8 {z_1}^5 - 4({z_2} + {z_3}) {z_1}^4 - 2 {\lambda_\infty} {z_1}^3 
				+ t {z_1}^2 + (2 {\lambda_\infty} {z_2}{z_3} + t {z_2} + t {z_3}) {z_1} - 3t {z_2}{z_3}}
			{({z_1} - {z_2})^3 ({z_1} - {z_3})^3 (2 {z_1}^3 - {\lambda_\infty} {z_1} + t)^2} \\ 
	&\qquad
		- \frac{{z_2}^2 (8 {z_2}^5 - 4({z_3} + {z_1}) {z_2}^4 - 2 {\lambda_\infty} {z_2}^3 
				+ t {z_2}^2 + (2 {\lambda_\infty} {z_3}{z_1} + t {z_3} + t {z_1}) {z_2} - 3t {z_3}{z_1}}
			{({z_2} - {z_3})^3 ({z_2} - {z_1})^3 (2 {z_2}^3 - {\lambda_\infty} {z_2} + t)^2} \\ 
	&\qquad
		- \frac{{z_3}^2 (8 {z_3}^5 - 4({z_1} + {z_2}) {z_3}^4 - 2 {\lambda_\infty} {z_3}^3 
				+ t {z_3}^2 + (2 {\lambda_\infty} {z_1}{z_2} + t {z_1} + t {z_2}) {z_3} - 3t {z_1}{z_2}}
			{({z_3} - {z_1})^3 ({z_3} - {z_2})^3 (2 {z_3}^3 - {\lambda_\infty} {z_3} + t)^2}
		\biggl\} \\
	&\quad
		\times d{z_1} \, d{z_2} \, d{z_3}, \\
	W_{1, 1}(z)
	&= - \frac{(4 z^3 - t) (8 {\lambda_\infty} z^4 - 20t z^3 + 2 {\lambda_\infty} t z - t^2)}
			{8(2 z^3 - {\lambda_\infty} z + t)^4} \, dz,
\end{align*}
\begin{align*}	
	F_0(\lambda_\infty, t)
	= - \frac{{\lambda_\infty}^2}{4} \log{(-2 t)}, \quad
	F_1(\lambda_\infty, t)
	= - \frac{1}{8} \log{t}.
\end{align*}
\begin{rem}
	It seems $W_{0,3}$ has singularities at  $z_1 = z_2 = z_3$, 
	but we can verify that $W_{0,3}$ is holomorphic there. 
\end{rem}

We choose 
\begin{equation}
\label{eq:(2,3)_D}
\begin{split}
	D(z ; \nu) 
	&= [z] - (1 - \nu_\infty) [0] - \nu_\infty [\infty] \\
	&= (1 - \nu_\infty) ([z] - [0]) + \nu_\infty ([z] - [\infty]) 
\end{split}
\end{equation}
as the divisor for the quantization. 
\begin{rem}
	$z = 0$ is a double pole of $x(z)$, i.e., $0 \in R$, but we can verify that $0 \notin R^{\ast}$. 
	Therefore, we can choose $\beta = 0$ as a base point. 
\end{rem}
Then, \thmref{WKB-Wg,n-BE} gives the quantum curve of the (2,3) curve (quantum (2,3) curve): 
\begin{equation}
\label{eq:(2,3)_eq(d/dx)}
	\left\{ 4 \hbar^3 \frac{d^3}{dx^3} - 2 x \hbar^2 \frac{d^2}{dx^2} 
			+ 2 ( \hat{\lambda}_\infty - \hbar ) \hbar \frac{d}{dx} - t \right\} \psi = 0, 
\end{equation}
where 
\begin{equation} \label{eq:lambda-hat-(2,3)}
	\hat{\lambda}_\infty = \lambda_\infty - \nu_{\infty} \hbar.
\end{equation}

\begin{lem}
\label{lem:(1,4)_Sn} 
For $m = 1, 2, \cdots$, we have
\begin{align} 
S_m(x, \lambda, \nu) = O(x^{-{3/2}}) 
\quad
(x \rightarrow \infty).
\end{align} 
\end{lem}


\section{Voros coefficients and the free energy}
\label{sec:Voros-vs-TR}


\subsection{Relations between Voros coefficients and the free energy}
\label{subsec:Voros-vs-TR}


In this subsection we formulate the main results which allow us to express the Voros coefficients 
of the quantum curves discussed in \S \ref{subsec:quantum-curve} by the free energy 
with a parameter shift. 

Let 
\begin{equation} 
\label{eq:total-free-energy}
	F({\lambda_{\infty}}, t; \hbar)
	= \sum_{g = 0}^{\infty} \hbar^{2g - 2} F_g({\lambda_{\infty}}, t)
\end{equation}
be the free energy for the spectral curve in \S \ref{subsec:quantum-curve}. 
Then, the precise statement is formulated as follows. 

\begin{thm}
\label{thm:main(i)}
\begin{equation} 
\label{eq:V-and-F-general}
	V({\lambda_{\infty}}, t, \nu_{\infty}; \hbar)
	= F(\hat{\lambda}_{\infty} + \hbar, t, \hbar) - F(\hat{\lambda}_{\infty}, t, \hbar) 
		- \frac{\partial F_0}{\partial \lambda_{\infty}} \hbar^{-1} 
		+ \frac{2 \nu_{\infty} - 1}{2} \frac{\partial^2 F_0}{\partial {\lambda_{\infty}}^2}. 
\end{equation}
Here $\hat{\lambda}_{\infty} = {\lambda_{\infty}} - \nu_{\infty} \hbar$ as we have introduced in \eqref{eq:lambda-hat-(1,4)}. 
\end{thm}

We can prove \thmref{main(i)} similarly to the case of the Weber equation because the proof of \thmref{main(i)} does not depend on $t$. 

To prove \thmref{main(i)}, we need the following identity. 

\begin{lem}
\label{lem:variation}
\begin{equation}
\label{eq:variation}
	\frac{\partial^n}{\partial{\lambda_{\infty}}^n} F_g
	= \int_{\zeta_1 = 0}^{\zeta_1=\infty}\cdots \int_{\zeta_n = 0}^{\zeta_n=\infty}
		W_{g, n}(\zeta_1, \cdots, \zeta_n)\qquad (2g + n \geq 3).
\end{equation}
\end{lem}

\begin{proof}[Proof of Lemma \ref{lem:variation}]
Because
\begin{equation}
	\Omega(z) 
	= \frac{\partial y(z)}{\partial {\lambda_{\infty}}} \cdot dx(z)
		- \frac{\partial x(z)}{\partial {\lambda_{\infty}}} \cdot dy(z)
	= - \frac{dz}{z}
	= \int^{\zeta = \infty}_{\zeta = 0} B(z, \zeta)
\end{equation}
holds, Theorem \ref{thm:VariationFormula} 
gives \eqref{eq:variation}, except for the case $g=0$. 
By using the expressions of $W_{0,3}$ and $F_0$, 
we can verify \eqref{eq:variation} holds for for $(g,n) = (0,3)$ directly. 
Therefore, thanks to Theorem \ref{thm:VariationFormula}, 
we can conclude that \eqref{eq:variation} is also valid for $g=0$ and $n \ge 3$. 
This completes the proof. 
\end{proof}

\begin{proof}[Proof of Theorem \ref{thm:main(i)}]
By Theorem \ref{thm:WKB-Wg,n-BE}, the Voros coefficient can be rewritten as
\begin{align}
	V({\lambda_{\infty}}, t, \nu_{\infty}; \hbar) 
	&= \sum_{m = 1}^{\infty} \hbar^m \int_0^\infty 
		\Bigl( S(x(z), {\lambda_{\infty}}, t, \nu_{\infty}; \hbar) - \hbar^{-1} S_{-1}(x(z), {\lambda_{\infty}}, t) 
			- S_0(x(z), {\lambda_{\infty}}, t, \nu_{\infty}) 
		\Bigr) \frac{dx}{dz} \, dz \\ 
	&= \sum_{m = 1}^{\infty} \hbar^m \int_0^\infty 
		\left\{ \sum_{\substack{2g + n - 2 = m \\ g \geq 0, \, n \geq 1}} 
				\frac{1}{n!} \frac{d}{dz} \int_{\zeta_1 \in D(z; \underline{\nu})}
				\cdots \int_{\zeta_n \in D(z; \underline{\nu})} W_{g, n}(\zeta_1, \ldots, \zeta_n) 
		\right\} dz \notag \\
	&= \sum_{m = 1}^{\infty} \hbar^m 
		\sum_{\substack{2g + n - 2 = m \\ g \geq 0, \, n \geq 1}} \frac{1}{n!} 
			\left( \int_{\zeta_1 \in D(\infty; \underline{\nu})} \cdots 
					\int_{\zeta_n \in D(\infty; \underline{\nu})} \right. \notag \\
	&\qquad\qquad\qquad\qquad\qquad\qquad \left. 
					- \int_{\zeta_1 \in D(0; \underline{\nu})} \cdots 
					\int_{\zeta_n \in D(0; \underline{\nu})} 
			\right) W_{g, n}(\zeta_1, \ldots, \zeta_n). \notag
\end{align}
Because
\begin{equation}
	D(\infty; \underline{\nu}) = (1 - \nu_{\infty}) ([\infty] - [0]) 
	\quad\text{and}\quad 
	D(0; \underline{\nu}) = - \nu_{\infty} ([\infty] - [0]), 
\end{equation}
we have
\begin{equation}
	V({\lambda_{\infty}}, t, \nu_{\infty}; \hbar) 
	= \sum_{m = 1}^{\infty} \hbar^m \sum_{\substack{2g + n - 2 = m \\ g \geq 0, \, n \geq 1}} 
		\frac{(1 - \nu_{\infty})^n - (- \nu_{\infty})^n}{n!} \int_0^\infty \cdots \int_0^\infty 
				W_{g, n}(\zeta_1, \ldots, \zeta_n). 
\end{equation}
Now we use Lemma \ref{lem:variation}:
\begin{align}
	V({\lambda_{\infty}}, t, \nu_{\infty}; \hbar) 
	&= \sum_{m = 1}^{\infty} \hbar^m \sum_{\substack{2g + n - 2 = m \\ g \geq 0, \, n \geq 1}} 
		\frac{(1 - \nu_{\infty})^n - (- \nu_{\infty})^n}{n!} \frac{ \partial^n F_g }{ \partial {\lambda_{\infty}}^n } \\
	&= \sum_{n = 1}^{\infty} \frac{(1 - \nu_{\infty})^n - (- \nu_{\infty})^n}{n!} 
		\hbar^n \frac{ \partial^n F({\lambda_{\infty}}, t; \hbar) }{ \partial {\lambda_{\infty}}^n } 
			- \frac{(1 - \nu_{\infty}) - (- \nu_{\infty})}{\hbar}\frac{\partial F_0}{\partial {\lambda_{\infty}}} 
		\notag\\
	&\qquad 
			- \frac{(1 - \nu_{\infty})^2 - (- \nu_{\infty})^2}{2!} \frac{\partial^2 F_0}{\partial{\lambda_{\infty}}^2} 
		\notag\\
	&= F \left({\lambda_{\infty}} - \nu_{\infty} \hbar + \hbar, t; \hbar \right) 
		- F \left({\lambda_{\infty}} - \nu_{\infty} \hbar, t; \hbar \right) 
		- \frac{\partial F_0}{\partial {\lambda_{\infty}}} \hbar^{-1}
		+ \frac{2 \nu_{\infty} - 1}{2} \frac{\partial^2 F_0}{\partial {\lambda_{\infty}}^2}. \notag 
\end{align}
\end{proof}

\begin{rem} \label{rem:regularization}
In the definition (\ref{eq:def-Voros-coeff}) of the Voros coefficient, we subtracted the first two terms 
$\hbar^{-1}S_{-1}$ and $S_0$ because these terms are singular at endpoints of the path 
$\gamma$. 
However, a regularization procedure of divergent integral (see \cite{Voros-zeta} for example) 
allows us to define the regularized Voros coefficient as follows:
\begin{equation} 
	V_{\rm reg}({\lambda_{\infty}}, t, \nu_{\infty}; \hbar) 
	:= \hbar^{-1}V_{-1}({\lambda_{\infty}}, t, \nu_{\infty}) + V_0({\lambda_{\infty}}, t, \nu_{\infty}) 
		+ V({\lambda_{\infty}}, t, \nu_{\infty}; \hbar), 
\end{equation}
where $V_{-1}({\lambda_{\infty}}, t, \nu_{\infty})$ and $V_0({\lambda_{\infty}}, t, \nu_{\infty})$ are obtained by solving 
\begin{equation} 
\label{eq:zeta-regularization-equation}
	\frac{\partial^2}{\partial {\lambda_{\infty}}^2} V_{-1} 
	= \int_{\gamma} \frac{\partial^2}{\partial {\lambda_{\infty}}^2} S_{-1}(x) \, dx, \quad 
	\frac{\partial}{\partial {\lambda_{\infty}}} V_{0}
 	= \int_{\gamma} \frac{\partial}{\partial {\lambda_{\infty}}} S_{0}(x) \, dx.
\end{equation}
Actually, we can verify that $\partial_{{\lambda_{\infty}}}^2 S_{-1}(x) dx$ and 
$\partial_{{\lambda_{\infty}}}S_0(x) dx$ are holomorphic at $x=\infty$ although  
$S_{-1}$ and $S_0$ are singular there.
Hence, the equations \eqref{eq:zeta-regularization-equation} make sense
and we can find $V_{-1}$ and $V_{0}$. For example, in the case of the (1,4) quantum curve, 
we obtain 
\begin{align}
	\frac{\partial^2}{\partial {\lambda_{\infty}}^2} V_{-1} 
	= \frac{1}{{\lambda_{\infty}}}, \qquad 
	\frac{\partial}{\partial {\lambda_{\infty}}} V_{0} 
	= - \frac{2 \nu_{\infty} - 1}{2 {\lambda_{\infty}}}. 
\end{align}
Actually, we can verify that the regularized integrals are realized by the correction terms 
\begin{equation}
	V_{-1} = \frac{\partial F_0}{\partial {\lambda_{\infty}}}, \qquad
	V_{0} = - \frac{2 \nu_{\infty} - 1}{2} \frac{\partial^2 F_0}{\partial {\lambda_{\infty}}^2} 
\end{equation}
in the right hand-side of the relation \eqref{eq:V-and-F-general}.
Thus we conclude that the regularized Voros coefficient satisfies 
\begin{equation} 
\label{eq:Vreg-and-free-energy}
	V_{\rm reg}({\lambda_{\infty}}, t, \nu_{\infty}; \hbar) 
	= F \left({\lambda_{\infty}} - \nu_{\infty} \hbar + \hbar, t; \hbar \right) 
		- F \left({\lambda_{\infty}} - \nu_{\infty} \hbar, t; \hbar \right). 
\end{equation}
\end{rem}



\subsection{Three-term difference equations satisfied by the free energy} 
\label{subsec:Three-term_difference-eq.} 


In this subsection, we derive the three-term difference equation which the generating function of the free energies satisfies. The precise statement is formulated as follows. 

\begin{thm}
\label{thm:main(ii)}
The free energy \eqref{eq:total-free-energy} satisfies the following difference equation.
\begin{equation}
\label{eq:free-energy_difference-eq.}
	F({\lambda_{\infty}} + \hbar, t; \hbar) - 2 F({\lambda_{\infty}}, t; \hbar) + F({\lambda_{\infty}} - \hbar, t; \hbar) 
	= \frac{\partial^2 F_0}{\partial {\lambda_{\infty}}^2}. 
\end{equation}
\end{thm}

We will only give the proof for the (quantum) (1,4) curve because the result for the (quantum) (2,3) curve is proved in a similar manner. 

To prove Theorem \ref{thm:main(ii)}, we need the following identity.

\begin{lem}
\label{lem:Voros-parameter}
\begin{equation}
\label{eq:Voros-difference}
	V({\lambda_{\infty}}, t, - \nu_{\infty}, \hbar) - V({\lambda_{\infty}}, t, 1 - \nu_{\infty}, \hbar) 
		= - \log{ \left( 1 - \frac{\nu_{\infty} \hbar}{{\lambda_{\infty}}} \right) }.
\end{equation}
\end{lem}

\begin{proof}[Proof of Theorem \ref{thm:main(ii)}]
From \lemref{Voros-parameter}, 
\begin{equation}
\label{eq:main:tmpeq}
	V({\lambda_{\infty}}, t, 0, \hbar) = V({\lambda_{\infty}}, t, 1, \hbar).
\end{equation}
It follows from \thmref{main(i)} that
\begin{align}
	V({\lambda_{\infty}}, t, 0, \hbar)
		&= F \left({\lambda_{\infty}} + \hbar, t; \hbar \right) - F \left({\lambda_{\infty}}, t; \hbar \right)
			- \frac{\partial F_0}{\partial {\lambda_{\infty}}} \hbar^{-1}
			- \frac{1}{2} \frac{\partial^2 F_0}{\partial {\lambda_{\infty}}^2}, \\
	V({\lambda_{\infty}}, t, 1, \hbar) 
		&= F \left({\lambda_{\infty}}, t; \hbar \right) - F \left({\lambda_{\infty}} - \hbar, t; \hbar \right)
			- \frac{\partial F_0}{\partial {\lambda_{\infty}}} \hbar^{-1}
			+ \frac{1}{2} \frac{\partial^2 F_0}{\partial {\lambda_{\infty}}^2}.
\end{align}
By substituting these two relations into \eqref{eq:main:tmpeq}, we obtain \thmref{main(ii)}.
\end{proof}


\subsection{The explicit form of the free energy}
\label{subsec:FreeEnergy} 


We obtain explicit formulas for the coefficients of the free energy and Voros coefficients. In this subsection we provide the explicit expressions for the free energy. We will only give the proof for the (quantum) (1,4) curve because the result for the (quantum) (2,3) curve is proved in a similar manner. 

\begin{thm}
\label{thm:main(iii)}
For $g \geq 2$, the $g$-th free energy of 
the spectral curve $(C)$
has the following expression.
\begin{itemize}
	\item[$\bullet$] For (1,4) curve (\S \ref{subsection:quantum-(1,4)}):
\end{itemize} \vspace{-1.3em}
\begin{equation}
\label{eq:(1,4)_Fg(concrete-form)}
	F_g({\lambda_{\infty}}, t) = \frac{B_{2g}}{2g(2g - 2)} \dfrac{1}{{{\lambda_\infty}}^{2g-2}} \quad (g \geq 2), 
\end{equation}
where $\{B_n\}_{n \geq 0}$ designates the Bernoulli number defined by 
\begin{equation}
\label{def:Bernoulli}
	\frac{w}{e^w - 1} = \sum_{n = 0}^{\infty} B_n \frac{w^n}{n!}.
\end{equation}
($F_0$ and $F_1$ for (1,4) curve are given in \S \ref{subsection:quantum-(1,4)}.) 
\begin{itemize}
	\item[$\bullet$] For (2,3) curve (\S \ref{subsection:quantum-(2,3)}):
\end{itemize} \vspace{-1.3em}
\begin{equation}
\label{eq:(2,3)_Fg(concrete-form)}
	F_g({\lambda_{\infty}}, t) = 0 \quad (g \geq 2).
\end{equation}
($F_0$ and $F_1$ for (2,3) curve are given in \S \ref{subsection:quantum-(2,3)}.) 
\end{thm}


To prove \thmref{main(iii)}, we need the following lemma. 

\begin{lem}
\label{lem:t-dependence}
\begin{equation}
\label{eq:t-dependence}
	\frac{ \partial F_g }{ \partial t} = 0 \qquad (g \geq 1).
\end{equation}
\end{lem}

Lemma \ref{lem:t-dependence} is obtained from 

\begin{lem}
\label{lem:variation-t}
For the (1,4) equation 
\begin{equation}
\label{eq:variation-t}
	\frac{ \partial F_{g} }{ \partial t} 
	= - \Res_{z = \infty} z^2 \, W_{g,1}(z) 
\end{equation}
holds. 
\end{lem}

\begin{lem}
\label{lem:variation-t()}
For the (1,4) equation the following relations hold: 
\begin{align}
\label{eq:variation-t()_1}
	\Res_{z = \infty} z^2 \sum_{m = -1}^{\infty} \hbar^m S_m(x(z)) dx(z) 
	= C_{-1}(z, {\lambda_{\infty}}, \nu_{\infty}) \hbar^{-1} + C_0(z, {\lambda_{\infty}}, \nu_{\infty}), \\
\label{eq:variation-t()_3}
	\Res_{z = \infty} z^2 \sum_{\substack{g \geq 0, \, n \geq 2 \\ (g, n) \ne (0, 2)}} 
		\frac{\hbar^{2g - 2 + n}}{(n-1)!} 
		\int_{\infty}^z \cdots \int_{\infty}^z W_{g, n}(z, z_2, \ldots, z_n)
	= \sum_{g \geq 1} \hbar^{2g} C_{g, 2}, 
\end{align}
where $C_{-1}$, $C_{0}$ and $C_{g, 2}$ $(g \geq 1)$ are constant with respect to $\hbar$. 
\end{lem}

\begin{proof}[Proof of Lemma \ref{lem:variation-t}]
By using the Riccati equation, 
we can verify \eqref{eq:variation-t()_1} directly. 
Because
\begin{equation}
	\Omega(z) 
	= \frac{\partial y(z)}{\partial t} \cdot dx(z)
		- \frac{\partial x(z)}{\partial t} \cdot dy(z)
	= 2z \, dz 
	= - \frac{1}{2 \pi i} \int_{\zeta \in \gamma} {\zeta}^2 B(z, \zeta)
\end{equation}
holds, Theorem \ref{thm:VariationFormula} 
gives \eqref{eq:variation-t}. 
\end{proof}

\begin{proof}[Proof of Lemma \ref{lem:variation-t()}]
Because $W_{g, n}(z_1, \cdots, z_n)$ are holomorphic at $z_i = \infty$ $(1 \leq i \leq n)$ 
for $2g - 2 + n \geq 1$, we find that 
\begin{align*}
	W_{g, n}(z, z_2, \ldots, z_n) 
	\sim \frac{d z_i}{{z_i}^2} \left( {C}^{(i)}_{g,n} + O(1/z_i) \right)
	\quad (z_i \rightarrow \infty ). 
\end{align*}
Then, since lower order terms $O(1/z_i)$ vanish in the limit $z_i \to \infty$ $(1 \leq i \leq n)$,
we obtain
\begin{align*}
	\int_{\zeta_2 = \infty}^{\zeta_2 = z_2} W_{g, n}(z, \zeta_2, \ldots, \zeta_n) 
	= \int_{\zeta_2 = \infty}^{\zeta_2 = z_2}  
			\frac{C_{g,n} d \zeta_2 \cdots d \zeta_n}
				{{z}^2 {\zeta_2}^2 {\zeta_3}^2 \cdots {\zeta_n}^2} d z
	= - \frac{C_{g,n} d \zeta_3 \cdots d \zeta_n}{{z}^2 {z_2} {\zeta_3}^2 \cdots {\zeta_n}^2} d z.
\end{align*}
Therefore, 
\begin{align*}
	\left. \int_{\infty}^{z_2} \cdots \int_{\infty}^{z_n} W_{g, n}(z, z_2, \ldots, z_n) 
	\right|_{z_2 = \cdots = z_n = z}
	\sim \left( \frac{(-1)^{n+1}C_{g,n}}{z^{n+1}} + \cdots \right) dz 
\end{align*}
holds. Multiplying both sides of the equation by $z^2$ and calculating residues, we obtain \eqref{eq:variation-t()_3}. 

\end{proof}

\begin{proof}[Proof of Lemma \ref{lem:t-dependence}]
By taking $\nu = 0$ in \thmref{WKB-Wg,n-BE} we obtain 
\begin{align}
	&\left. \log{\psi} \right|_{x = x(z)} = \sum_{m = -1}^{\infty} \hbar^m \int^{x(z)} S_m dx \\
	&= \sum_{m = -1}^{\infty} \hbar^m 
	  	\left\{ \sum_{\substack{2g + n - 2 = m \\ g \geq 0, \, n \geq 1}} 
			\frac{1}{n!} \int_{\infty}^z \cdots \int_{\infty}^z 
				\left( W_{g, n}(z_1, \ldots, z_n) 
						- \delta_{g,0} \delta_{n,2} \frac{dx(z_1) \, dx(z_2)}{(x(z_1) - x(z_2))^2} 
				\right)
		\right\}. \notag 
\end{align}
It follows from this equation that 
\begin{align}
	\sum_{g \geq 0} \hbar^{2g - 1} \left(- \Res_{z = \infty} z^2 W_{g, 1}(z) \right) 
	&= - \Res_{z = \infty} z^2 \sum_{m = -1}^{\infty} \hbar^m S_m(x(z)) dx(z) \\
	&\qquad \notag 
		+ \Res_{z = \infty} z^2 
			\int_{\infty}^z \left( W_{0, 2}(z, z_2) - \frac{dx(z) \, dx(z_2)}{(x(z) - x(z_2))^2} \right) \\
	&\qquad \notag
		+ \Res_{z = \infty} z^2 \sum_{\substack{g \geq 0, \, n \geq 2 \\ (g, n) \ne (0, 2)}} 
			\frac{\hbar^{2g - 2 + n}}{(n-1)!} 
			\int_{\infty}^z \cdots \int_{\infty}^z W_{g, n}(z, z_2, \ldots, z_n). 
\end{align}
Because the left hand side of this equation is written by 
\begin{align*}
	\sum_{g \geq 0} \hbar^{2g - 1} \left(- \Res_{z = \infty} z^2 W_{g, 1}(z) \right) 
	= - \hbar^{-1} \Res_{z = \infty} z^2 W_{0, 1}(z) 
		+ \sum_{g \geq 1} \hbar^{2g - 1} \frac{ \partial F_g }{ \partial t}, 
\end{align*}
we compare the odd terms with respect to $\hbar$ of both sides. 
By using \lemref{variation-t()} we find that there is no odd term whose order with respect to 
$\hbar$ is greater than or equal to one in the right hand side. 
It means that \eqref{eq:t-dependence} holds. 
\end{proof}

Now we give a proof of \thmref{main(iii)}.

\begin{proof}[Proof of Theorem \ref{thm:main(iii)}]
By using a shift operator (or an infinite order differential operator) 
$e^{\hbar\partial_{{\lambda_{\infty}}}}$, the equation (\ref{eq:free-energy_difference-eq.}) 
in \thmref{main(ii)} becomes
\begin{equation}
\label{prop:difference-eq:sol:tmp:1}
	e^{-\hbar\partial_{{\lambda_{\infty}}}} (e^{\hbar\partial_{{\lambda_{\infty}}}} - 1)^2 F({\lambda_{\infty}}, t; \hbar)
	= \frac{\partial^2 F_0}{\partial {\lambda_{\infty}}^2}. 
\end{equation}
It follows from 
\begin{equation}
	e^{-w} (e^w - 1)^2 
	\left\{ \frac{1}{w^2} - \sum_{n = 0}^{\infty} \frac{B_{n + 2}}{\, n + 2 \,} \frac{\, w^n \,}{\, n! \,}
	\right\} = 1
\end{equation} 
(which follows from the definition \eqref{def:Bernoulli} of the Bernoulli numbers) that
\begin{equation}
	e^{-\hbar\partial_{{\lambda_{\infty}}}} (e^{\hbar\partial_{{\lambda_{\infty}}}} - 1)^2
	\left\{ (\hbar\partial_{{\lambda_{\infty}}})^{-2} - \sum_{n = 0}^{\infty} \frac{B_{n + 2}}{\, n + 2 \,} 
			\frac{\, (\hbar\partial_{{\lambda_{\infty}}})^n \,}{\, n! \,}
	\right\} = {\rm{id}}.
\end{equation}
Hence we find that
\begin{align}
\label{sol:FreeEnergy}
	\hat{F}({\lambda_{\infty}}, t;\hbar)
	&:= \left\{ (\hbar\partial_{{\lambda_{\infty}}})^{-2} - \sum_{n = 0}^{\infty} \frac{B_{n + 2}}{\, n + 2 \,} 
				\frac{\, (\hbar\partial_{{\lambda_{\infty}}})^n \,}{\, n! \,}
		\right\} \frac{\partial^2 F_0}{\partial {\lambda_{\infty}}^2} \\
	&= \hbar^{-2} F_0({\lambda_{\infty}}, t) 
		- \frac{1}{12} \frac{\partial^2 F_0}{\partial {\lambda_{\infty}}^2} 
		+ \sum_{g = 2}^{\infty} \frac{B_{2g}}{2g(2g-2)} \frac{\hbar^{2g - 2}}{{\lambda_{\infty}}^{2g-2}} 
		+ \hat{F}_t (t)
		\notag
\end{align}
is a solution of \eqref{eq:free-energy_difference-eq.}. 
Here we note that, 
\begin{equation}
\frac{\partial^2 F_0}{\partial {\lambda_{\infty}}^2} = \frac{1}{2} \log{(-3 {\lambda_{\infty}}^2)} 
\end{equation}
holds. 

Since $F$ and $\hat{F}$ satisfies the same difference equation 
\eqref{eq:free-energy_difference-eq.}, their difference 
	$G := F - \hat{F} = \sum_{g=2}^{\infty} \hbar^{2g-2} G_{g}({\lambda_{\infty}}, t)$ 
satisfies 
\begin{equation}
	G({\lambda_{\infty}} + \hbar, t; \hbar) - 2G({\lambda_{\infty}}, t; \hbar) + G({\lambda_{\infty}} - \hbar, t; \hbar) = 0.
\end{equation}
This relation implies that, each coefficient $G_{g}({\lambda_{\infty}}, t)$ of $G$ must satisfy 
$\partial_{\lambda_{\infty}}^2 G_{g} = 0$.
Therefore, each term of $G$ must be a linear in ${\lambda_{\infty}}$. 
However, due to the homogeneity 
and \lemref{t-dependence}, 
$F_g - \hat{F}_g$ must be zero for all $g$. 
This shows the desired equality \eqref{eq:(1,4)_Fg(concrete-form)}. 
\end{proof}


\subsection{The explicit form of Voros coefficients} 
\label{subsec:voros} 

 
In this subsection we provide the explicit expressions for Voros coefficients. We will only give the proof for the (quantum) (1,4) curve because the result for the (quantum) (2,3) curve is proved in a similar manner. 

\begin{thm}
\label{thm:main(iv)}
The Voros coefficients 
for the following quantum curve 
has the following expression. 
\begin{itemize}
	\item[$\bullet$] For (1,4) curve (\S \ref{subsection:quantum-(1,4)}):
\end{itemize} \vspace{-1.3em}
\begin{equation}
\label{eq:(1,4)_Voros(concrete-form)}
	V({\lambda_{\infty}}, t, \nu_{\infty}, \hbar) 
	= \sum_{m = 1}^{\infty} \frac{B_{m+1}(\nu_{\infty})}{m(m + 1)} 
		\left( \frac{\hbar}{{\lambda_{\infty}}} \right)^{m}.
\end{equation}
Here $B_m(t)$ is the Bernoulli polynomial defined through the generating function as
\begin{equation}
\label{def:BernoulliPoly}
\frac{w e^{X w}}{e^w - 1} = \sum_{m = 0}^{\infty} B_m(X) \frac{w^m}{m!}.
\end{equation}
(These expressions were also obtained in \cite{IKo}.)
\begin{itemize}
	\item[$\bullet$] For (2,3) curve (\S \ref{subsection:quantum-(2,3)}):
\end{itemize} \vspace{-1.3em}
\begin{equation}
\label{eq:(2,3)_Voros(concrete-form)}
	V({\lambda_{\infty}}, t, \nu_{\infty}, \hbar) = 0. 
\end{equation}
\end{thm}


%

\begin{proof}
The relation \eqref{eq:Vreg-and-free-energy} 
between the regularized Voros coefficient 
and the free energy can be written as
\begin{equation}
	V_{\rm reg}(\lambda_\infty, t, \nu_\infty; \hbar) 
	= e^{- \nu_\infty \hbar \partial_{\lambda_\infty}} \Big( e^{\hbar\partial_{\lambda_\infty}} - 1 \Big) 
		F(\lambda_\infty, t; \hbar) 
\end{equation}
by the shift operators. 
Using the three term relation \eqref{eq:free-energy_difference-eq.} of $F$, 
we have
\begin{equation} 
\label{eq:(1,4)-diffrence-eq-for-V}
\begin{split}
	e^{(\nu_\infty - 1) \hbar \partial_{\lambda_\infty}} \Big( e^{\hbar\partial_{\lambda_\infty}} - 1 \Big) 
		V(\lambda_\infty, t, \nu_\infty; \hbar)
	= e^{-\hbar \partial_{\lambda_\infty}} \Big( e^{\hbar\partial_{\lambda_\infty}} - 1 \Big)^2 
		F(\lambda_\infty, t; \hbar) 
	= \frac{1}{2} \log{(-3 {\lambda_\infty}^2)}.
\end{split}
\end{equation}

Let us invert the shift operator 
$e^{(\nu_\infty - 1) \hbar \partial_{\lambda_\infty}} \left( e^{\hbar\partial_{\lambda_\infty}} - 1 \right)$
(or solving the difference equation) 
to obtain an expression of $V_{\rm reg}$. 
For the purpose, we use a similar technique used in the previous subsection. 
Namely, it follows from 
\begin{equation}
	e^{- X w} (e^{w} - 1) 
	\left(\frac{1}{w} + \sum_{m=0}^{\infty} \frac{B_{m+1}(X)}{m+1} \frac{w^m}{m!} \right) = 1
\end{equation}
(cf.\,\eqref{def:BernoulliPoly}) that 
\begin{equation}
	e^{- X \hbar \partial_{\lambda_\infty}} (e^{\hbar \partial_{\lambda_\infty}} - 1) 
	\left( (\hbar \partial_{\lambda_\infty})^{-1} 
			+ \sum_{m=0}^{\infty} \frac{B_{m+1}(X)}{m+1} \frac{(\hbar \partial_{\lambda_\infty})^m}{m!} 
	\right) 
	= {\rm id}. 
\end{equation}
The last equality with $X = 1 - \nu_\infty$ shows that the formal series
\begin{align} 
\label{eq:expression-Vreg}
	V_{\rm reg} 
	&= \hbar^{-1} \frac{\partial F_0}{\partial \lambda_\infty} 
		- \frac{\nu_\infty - 1}{2} \frac{\partial^2 F_0}{\partial {\lambda_\infty}^2} 
		+ \sum_{m=1}^{\infty} \frac{B_{m+1}(1- \nu_\infty)}{m+1} 
		\frac{(\hbar \partial_{\lambda_\infty})^m \log \lambda_\infty}{m!} \\
		\notag 
	&= \hbar^{-1} V_{-1} + V_0 + 
		\sum_{m=1}^{\infty} \frac{(-1)^{m+1}B_{m+1}(1 - \nu_\infty)}{m(m+1)} 
		\left(\frac{\hbar}{\lambda_\infty}\right)^m \\
		\notag 
	&= \hbar^{-1} V_{-1} + V_0 + 
		\sum_{m=1}^{\infty} \frac{B_{m+1}(\nu_\infty)}{m(m+1)} 
		\left(\frac{\hbar}{\lambda_\infty}\right)^m 
\end{align}
satisfies the difference equation \eqref{eq:(1,4)-diffrence-eq-for-V}.
Here we used $B_1(X) = X - 1/2$ and the equality 
$B_m(X) = (-1)^{m}B_m(1-X)$.
\end{proof}




\end{document}